\documentclass[11pt]{article}

\usepackage{microtype}
\usepackage{xifthen}
\usepackage{complexity}
\usepackage{tikz}
\usepackage[ruled,vlined]{algorithm2e}
\usepackage{amsmath}
\usepackage{amsthm}
\usepackage{fullpage}

\renewcommand{\SAT}[1][]{\ifthenelse{\isempty{#1}}{}{#1\textsf{-}}\mathsf{SAT}}
\newcommand{\sSAT}[1][]{\#\SAT[#1]}
\newcommand{\sP}{\#\P}
\newcommand{\calH}{\mathcal{H}}
\newcommand{\calF}{\mathcal{F}}
\newcommand{\calG}{\mathcal{G}}
\newcommand{\calP}{\mathcal{P}}
\newcommand{\Ginc}{\calG_{\mathsf{inc}}}
\newcommand{\MIMw}{\mathsf{MIM}\text{-width}}

\newcommand{\mimw}{\mathbf{mimw}}
\newcommand{\decDNNF}{\mathsf{dec}\text{-}\mathsf{DNNF}}
\newcommand{\DNNF}{\mathsf{DNNF}}
\newcommand{\NNF}{\mathsf{NNF}}
\newcommand{\CNF}[1][]{\ifthenelse{\isempty{#1}}{}{#1\textsf{-}}\mathsf{CNF}}
\newcommand{\DNF}[1][]{\ifthenelse{\isempty{#1}}{}{#1\textsf{-}}\mathsf{DNF}}
\newcommand{\var}{\mathsf{var}}
\newcommand{\size}{\textsf{size}}
\newcommand{\sat}{\mathsf{sat}}

\newcommand{\Vle}[1]{V_{\leq #1}}
\newcommand{\Vlt}[1]{V_{<#1}}
\newcommand{\Vgt}[1]{V_{> #1}}
\newcommand{\Vge}[1]{V_{\ge #1}}

\newtheorem{theorem}{Theorem}
\newtheorem{lemma}[theorem]{Lemma}
\newtheorem{corollary}[theorem]{Corollary}
\newtheorem{definition}[theorem]{Definition}

\begin{document}
\title{Understanding the complexity of $\sSAT$ using knowledge compilation\footnote{This work was partially supported by ANR AGGREG.}}

\author{Florent Capelli \\ Birkbeck College, University of London \\ \small{florent@dcs.bbk.ac.uk}}
\date{\today}

\maketitle

\begin{abstract}
Two main techniques have been used so far to solve the $\sP$-hard problem $\sSAT$. The first one, used in practice, is based on an extension of DPLL for model counting called exhaustive DPLL. The second approach, more theoretical, exploits the structure of the input to compute the number of satisfying assignments by usually using a dynamic programming scheme on a decomposition of the formula. In this paper, we make a first step toward the separation of these two techniques by exhibiting a family of formulas that can be solved in polynomial time with the first technique but needs an exponential time with the second one. We show this by observing that both techniques implicitely construct a very specific boolean circuit equivalent to the input formula. We then show that every $\beta$-acyclic formula can be represented by a polynomial size circuit corresponding to the first method and exhibit a family of $\beta$-acyclic formulas which cannot be represented by polynomial size circuits corresponding to the second method. This result shed a new light on the complexity of $\sSAT$ and related problems on $\beta$-acyclic formulas. As a byproduct, we give new handy tools to design algorithms on $\beta$-acyclic hypergraphs.
\end{abstract}

\section{Introduction}

The problem $\sSAT$ of counting the satisfying assignments of a given $\CNF$-formula is a central problem to several areas such as probabilistic reasoning~\cite{Roth96,BacchusDP03} and probabilistic databases~\cite{BeameLRS14,BeameLRS13,JhaS13}. This problem is much harder than $\SAT$, its associated decision problem. For example, the problem $\SAT[2]$ of deciding if a formula having at most two literals per clause if satisfiable is easy where counting those satisfying assignments is as hard as $\sSAT$. Even computing a $2^{n^{1-\epsilon}}$-approximation in the restricted case of monotone $\SAT[2]$ is hard for any $\epsilon > 0$~\cite{Roth96}. 

In order to tackle this problem, two main approaches have been used so far. The first approach -- applied in practical tools for solving $\sSAT$ -- follows the successful road paved by $\SAT$-solvers: it is based on a variation of DPLL~\cite{DavisP60} called {\em exhaustive DPLL}~\cite{HuangD05} and the approach is mainly focused on improving the heuristics used for eliminating variables and choosing which subformulas should be cached during the computation. The performance of such tools -- though impressive for such a hard problem~\cite{HuangD05,sang04,thurley06,BacchusDP03} -- lag far behind the state-of-the-art $\SAT$-solvers. This gap is mainly explained by the differences between the hardness of both problems, but also by the fact that optimizations for exhaustive DPLL are inspired by those used in $\SAT$-solvers and not always relevant for model counting~\cite{SangBK05}.
The second -- more theoretical -- approach focuses on {\em structural restrictions} of the input formula.  The main idea of this approach is to solve $\sSAT$ more quickly on formulas where interaction between the clauses and the variables is restricted. This interaction is usually represented by a graph derived from the input $\CNF$-formula. The complexity of $\sSAT$ is then studied on inputs where the associated graph belongs to a restricted class of graphs. Samer and Szeider~\cite{SamerS10} were the first to formalize this idea for $\sSAT$ by showing that if this graph is of bounded tree width, then $\sSAT$ can be solved in polynomial time. This result has then been improved and completed by different work showing the tractability of $\sSAT$ for more general or incomparable classes of formulas~\cite{PaulusmaSlivovskySzeider16,SlivovskyS13,SaetherTV14,CapelliDM14}, the intended goal being to understand the frontier of tractibility for $\sSAT$.

\paragraph{Contributions.} The main contribution of this paper is to propose a formal framework, using tools from {\em knowledge compilation}, to study both algorithmic techniques and to compare their respective power. We then make a first step toward the separation of both techniques by exhibiting a class of formulas having the following property: for every formula $F$ of this class, there exists an elimination order of the variables for which exhaustive DPLL returns the number of satisfying assignments of $F$ in linear time while algorithms based on structural restrictions needs exponential time. 

The class of formulas we are using to separate both technique are $\beta$-acyclic formulas, a class already known to be tractable~\cite{BraultCM15}. The algorithm used to solve this class was however very different from the one that are usually used by structure-based algorithms. Our result gives a formal explanation of why the usual techniques fail on this class, a question that has puzzled the community since $\SAT$ has been shown tractable on this class of formulas without generalizing to counting~\cite{OrdyniakPS13}. 

Moreover, in Section~\ref{sec:comp}, we give tools that are useful for designing algorithms on $\beta$-acyclic hypergraphs and are of independent interest.


\paragraph{Methodology.} It has been observed that the trace of every implementation of exhaustive DPLL actually constructs a very specific Boolean circuit equivalent to the input formula~\cite{HuangD05}. Such circuits are known in knowledge compilation under the name of {\em decision Decomposable Negation Normal Form} ($\decDNNF$). We first show in Section~\ref{sec:comp} that $\beta$-acyclic formulas can be represented by linear size $\decDNNF$, which can be interpreted as the fact that exhaustive DPLL may solve this class of formula in polynomial time, if it chooses the right order to eliminate variables and the right caching methods.

Similarly, all structure-based algorithms for $\sSAT$ use the same kind of dynamic programming scheme and it has been shown that they all implicitly construct a very specific Boolean circuit equivalent to the input formula~\cite{BovaCMS15}. Such circuits are known under the name of {\em structured deterministic} $\DNNF$. We start by arguing in Section~\ref{sec:prel} that every algorithm using techniques similar to the one used by structure-based algorithms will implicitly construct a circuit having a special property called {\em determinism}. In Section~\ref{sec:deviation}, We exhibit a class of $\beta$-acyclic formulas having no polynomial size equivalent {\em structured $\DNNF$}, thus separating both methods.

\paragraph{Related work.} The class of $\beta$-acyclic formulas we use to prove the separation have already been shown to be tractable for $\sSAT$ and not tractable to the state-of-the-art structure-based algorithms~\cite{BraultCM15} but this result does not rule out the existence of a more general algorithm based on the same technique and solving every known tractable class. Our result is sufficiently strong to rule out the existence of such an algorithm.

New lower bounds have been recently shown for circuits used in knowledge compilation~\cite{BovaCMS16,BeameLRS14,BeameLRS13,BeameL15,DarwicheP10}. Moreover, knowledge compilation has already been used to prove limits of algorithmic techniques in the context of model counting. Beame et al.~\cite{BeameLRS13} for example have exhibited a very interesting class of queries on probabilistic databases that can be answered in polynomial time by using specific techniques but that cannot be represented by circuits corresponding to exhaustive DPLL. They conclude that solving the query by using well-known reduction to $\sSAT$ and then calling a $\sSAT$-solver is weaker than using their technique. Our result uses somehow the same philosophy but on a different algorithmic technique.


\paragraph{Organization of the paper.} The paper is organized as follows: Section~\ref{sec:prel} contains the needed definitions and concepts used through the paper. Section~\ref{sec:comp} describes the algorithm to transform $\beta$-acyclic formulas into circuits corresponding to the execution of an exhaustive DPLL algorithm. Finally, Section~\ref{sec:deviation} contains the formalization of the framework for studying algorithms based on dynamic programming along a branch decomposition and a proof that the $\beta$-acyclic case is not covered by this framework.

\section{Preliminaries}
\label{sec:prel}

\subsection{$\CNF$-formulas.} A \emph{literal} is a variable $x$ or a negated variable $\neg x$. A \emph{clause} is a finite set of literals. A clause is \emph{tautological} if it contains the same variable negated as well as unnegated. A \emph{(CNF) formula} (or \emph{CNF}, for short) is a finite set of non-tautological clauses. If $x$ is a variable, we let $\var(x) = \var(\neg x) = x$. Given a clause $C$, we denote by $\var(C) = \bigcup_{\ell \in C} \var(\ell)$ and given a $\CNF$-formula, we denote by $\var(F) = \bigcup_{C \in F} \var(C)$. The {\em size} of a $\CNF$-formula $F$, denoted by $\size(F)$, is defined to be $\sum_{C \in F} |\var(C)|$. A $\CNF$-formula is {\em monotone} if it does not contain negative literals.

Let $X$ be a set of variables. An {\em assignment} $\tau$ of $X$ is a mapping from $X$ to $\{0,1\}$. The set of assignments of $X$ is denoted by $\{0,1\}^X$. Given an assignment $\tau$ of $X$ and $X' \subseteq X$, we denote by $\tau|_{X'}$ the restriction of $\tau$ on $X'$. Given two sets $X,X'$, $\tau \in \{0,1\}^X$ and $\tau' \in \{0,1\}^{X'}$, we denote by $\tau \simeq \tau'$ if $\tau|_{X \cap X'} = \tau'|_{X \cap X'}$. If $\tau \simeq \tau'$, we denote by $\tau \cup \tau'$ the assignment of $X \cup X'$ such that for all $x \in X$, $(\tau \cup \tau')(x) = \tau(x)$ and for all $x \in X'$, $(\tau \cup \tau')(x) = \tau'(x)$.  

A {\em boolean function} $f$ on variables $X$ is a mapping from $\{0,1\}^X$ to $\{0,1\}$. We denote by $\tau \models f$ if $\tau \in \{0,1\}^X$ is such that $f(\tau)=1$ and by $\sat(f) = \{\tau \in \{0,1\}^X \mid \tau \models f \}$. Given $Y \subseteq X$ and $\tau \in \{0,1\}^Y$, we denote by $f[\tau]$ the boolean function on variables $X \setminus Y$ defined by for every $\tau' \in \{0,1\}^{X \setminus Y}$, $f[\tau](\tau') = f(\tau \cup \tau')$. 

A $\CNF$-formula $F$ naturally induces a boolean function. Extending assignments to literals in the usual way, we say that an assignment $\tau$ \emph{satisfies} a clause $C$ if there is a literal $\ell \in C$ such that $\tau(\ell) = 1$. An assignment \emph{satisfies} a formula $F$ if it satisfies every clause $C \in F$. In this paper, we often identify the $\CNF$-formula and its underlying boolean function. Thus, given a $\CNF$-formula $F$ on variables $X$ and an assignment $\tau$ of $Y \subseteq X$, we will use the notation $F[\tau]$ in the same way as for any other boolean function. Observe that $F[\tau]$ is still represented by a $\CNF$-formula of size less than $\size(F)$: it is the $\CNF$-formula where we have removed satisfied clauses from $F$ and removed the variables of $Y$ in each remaining clause. 

\subsection{Graphs and branch decompositions.} We assume the reader is familiar with the basics of graph theory. An introduction to the topic can be found in~\cite{Diestel05}. Given a graph $G = (V,E)$, we often denote by $V(G)$ the set of vertices of $G$ and by $E(G)$ the set of edges of $G$ if they have not been named explicitly before. $G$ is said {\em bipartite} if there exists a partition $V_1 \uplus V_2$ of $V$ such that for every $e \in E$, one end-point of $e$ is in $V_1$ and the other is in $V_2$. Given a graph $G = (V,E)$ and $X,Y \subseteq V$, we denote by $G[X,Y] = (V',E')$ the bipartite graph such that $V' = X \cup Y$ and $E' = \{\{u,v\} \in E \mid u \in X, v \in Y\}$.  An {\em induced matching} $M$ is a matching of $G$ such that for every $e,f \in M$, if $e = \{u,v\}$ and $f = \{u',v'\}$, we have  $\{u,u'\} \notin E$, $\{u,v'\} \notin E$, $\{v,u'\} \notin E$ and $\{v,v'\} \notin E$.

A {\em branch decomposition} of $G$ is a binary rooted tree $T$ whose leaves are in one-to-one correspondence with $V$. Given $t$ a vertex of $T$, we denote by $T_t$ the subtree of $T$ rooted in $t$. We denote by $V_t$ the set of leaves of $T_t$. 

The {\em maximal induced matching width}~\cite{Vatshelle12}, $\MIMw$ for short, of a vertex $t$ of $T$ is the size of the largest induced matching $M$ of $G[V \setminus V_t, V_t]$. The $\MIMw$ of $T$, denoted as $\mimw(T)$, is the maximal $\MIMw$ of its vertices. The {\em maximal induced matching width} of a graph $G$, denoted as $\mimw(G)$, is the minimal $\MIMw$ of all branch decomposition of $G$. Figure~\ref{fig:bd} depicts a graph $G$ together with a branch decomposition of $G$. The distinguished node $t$ of this branch decomposition has $\MIMw$ $1$ as the biggest induced matching of $G[V_t,V \setminus V_t]$ is of size one because the matching $\{\{1,4\},\{2,3\}\}$ is not induced.
\begin{figure}
  \centering
  \begin{tikzpicture}[scale=0.6]
    \node (v1) at (-1,-1) {$1$};
    \node (v2) at (1,-1) {$2$};
    \node (v3) at (1,1) {$3$};
    \node (v4) at (-1,1) {$4$};

    \draw (v1) -- (v2) -- (v3) -- (v4) -- (v1) --(v3);
  \end{tikzpicture}~~
  \begin{tikzpicture}[scale=0.6]
    \node (v1) at (-2,-1) {$1$};
    \node (v2) at (-1,-1) {$2$};
    \node (v3) at (0,-1) {$3$};
    \node (v4) at (1,-1) {$4$};
    \node (t1) at (-1.5,0) {$t$};
    \draw (v1) -- (t1) -- (v2); \draw (t1) -- (-0.5,1) -- (0.5,0) -- (v4) -- (0.5,0) --(v3);
  \end{tikzpicture}~~
  \begin{tikzpicture}[scale=0.6]
    \node (v1) at (-1,-1) {$1$};
    \node (v2) at (1,-1) {$2$};
    \node (v3) at (1,1) {$3$};
    \node (v4) at (-1,1) {$4$};

    \draw (v1) -- (v3); \draw (v1) -- (v4); \draw  (v2) --(v3);
  \end{tikzpicture}
  \caption{From left to right: a graph $G=(V,E)$, a branch decomposition of $G$ and $G[V_t, V \setminus V_t]$}
  \label{fig:bd}
\end{figure}
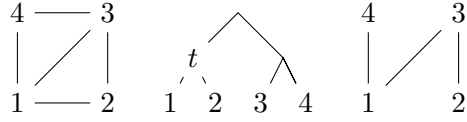

\subsection{Hypergraphs and $\beta$-acyclicity.} A {\em hypergraph} $\calH$ is a finite set of finite sets, called {\em edges}. We denote by $V(\calH) = \bigcup_{e \in \calH} e$ the set of vertices of hypergraph $\calH$.

Most notions on graphs may be naturally generalized to hypergraph. A hypergraph $\calH'$ is a subhypergraph of $\calH$ if $\calH' \subseteq \calH$. Given $S \subseteq V(\calH)$, we denote by $\calH \setminus S = \{e \setminus S \mid e \in \calH\}$. A {\em walk} of length $n$ from edge $e \in \calH$ to $f \in \calH$ is a sequence $(e_1,x_1, \dots, x_n, e_{n+1})$ of vertices and edges such that: $e = e_0$, $f = e_{n+1}$ and for every $i \leq n$, $x_i \in e_i \cap e_{i+1}$. A {\em path} is a walk that never goes twice through the same vertex nor the same edge. It is easy to check that if there is a walk from $e$ to $f$ in $\calH$, then there is also a path from $e$ to $f$.

There exist several generalizations of acyclicity to hypergraph introduced by Fagin~\cite{Fagin83} in the context of database query answer. An extensive presentation of hypergraph acyclicity notions may be found in~\cite{BraultB14}. In this paper, we focus on the $\beta$-acyclicity, which is the most general of such notions for which $\sSAT$ is still tractable. A hypergraph $\calH$ is $\beta$-acyclic if there exists an order $(x_1,\dots,x_n)$ of $V(\calH)$ such that for all $i \leq n$, for all $e,f \in \calH$ such that $x_i \in e \cap f$, then either $e \setminus \{x_1,\dots, x_i\} \subseteq f$ or $f \setminus \{x_1,\dots, x_i\} \subseteq e$. Such an order is called a $\beta$-elimination order of $\calH$. A $\beta$-acyclic hypergraph can be found on Figure~\ref{fig:hexdec}. The order $\{1,2,3,4,5\}$ is an $\beta$-elimination order.

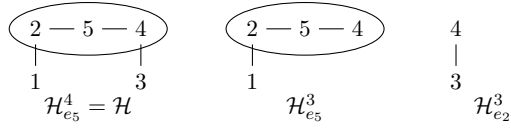
\begin{figure}
  \centering
\begin{tikzpicture}[scale=0.7, every node/.style={scale=0.8}]
	\draw (0,1) ellipse (1.5 and 0.5);
	\node (x5) at (0,1) {$5$};
	\node (x3) at (1,0) {$3$};
	\node (x4) at (1,1) {$4$};
	\node (x2) at (-1,1) {$2$};
	\node (x1) at (-1,0) {$1$};
	\draw (x1) -- (x2) -- (x5) -- (x4) -- (x3);
	\node (H45) at (0,-0.5) {$\mathcal{H}_{e_5}^4 = \calH$};
\end{tikzpicture} ~~~~
\begin{tikzpicture}[scale=0.7, every node/.style={scale=0.8}]
	\draw (0,1) ellipse (1.5 and 0.5);
	\node (x5) at (0,1) {$5$};
	\node (x4) at (1,1) {$4$};
	\node (x2) at (-1,1) {$2$};
	\node (x1) at (-1,0) {$1$};
	\draw (x1) -- (x2) -- (x5) -- (x4);
	\node (H45) at (0,-0.5) {$\mathcal{H}_{e_5}^3$};
\end{tikzpicture} ~~~~
\begin{tikzpicture}[scale=0.7, every node/.style={scale=0.8}]
	\node (x3) at (0,0) {$3$};
	\node (x4) at (0,1) {$4$};
	\draw (x4) -- (x3);
	\node (H35) at (0.7,-0.5) {$\mathcal{H}_{e_2}^3$};
\end{tikzpicture}	
  \caption{An example of $\calH_e^x$}
  \label{fig:hexdec}
\end{figure}

Given a hypergraph $\calH$, the incidence graph of $\calH$ is defined as the bipartite graph whose vertices are $V(\calH) \cup \calH$ and there is an edge between $e \in \calH$ and $x \in V(\calH)$ if and only if $x \in e$. The {\em incidence $\MIMw$} of a hypergraph is the $\MIMw$ of its incidence graph. The incidence $\MIMw$ of $\beta$-acyclic hypergraphs can be very large:

\begin{theorem}[\cite{BraultCM15}] 
  \label{thm:betavsmm} There exists an infinite family of $\beta$-acyclic hypergraphs of incidence $\MIMw$ $\Omega(n)$ where $n$ is the number of vertices of the hypergraph.
\end{theorem}

\subsection{Structure of formulas.} Let $F$ be a $\CNF$-formula. The {\em incidence graph} of $F$, denoted by $\Ginc(F)$, is the bipartite graph whose vertices are the variables and the clauses of $F$ and there is an edge between a variable $x$ and a clause $C$ if and only if $x \in \var(C)$. The {\em incidence $\MIMw$} of a formula $F$ is the $\MIMw$ of $\Ginc(F)$. The hypergraph of $F$, denoted by $\calH(F)$, is defined as $\calH(F) = \{\var(C) \mid C \in F\}$. A $\CNF$-formula is said to be {\em $\beta$-acyclic} if and only if its hypergraph is $\beta$-acyclic.

\subsection{Knowledge compilation} 

\paragraph{DNNF.}In this paper we focus on so-called $\DNNF$ introduced by Darwiche~\cite{Darwich01}. An extensive presentation of different target languages with their properties may be found in \cite{DarwicheM2002}. A Boolean circuit $C$ on variables $X$ is in {\em Negation Normal Form}, $\NNF$ for short, if its input are labeled by literals on $X$ and its internal gates are labeled with either a $\wedge$-gate or a $\vee$-gate. We assume that such circuit has a distinguished gate called the {\em output}. An $\NNF$ circuit $D$ computes the boolean function computed by its output gate and we will often identify the circuits and its computed Boolean function. We denote by $\size(D)$ the number of gates of $D$ and by $\var(D)$ the set of variables labeling its input. If $v$ is a gate of $D$, we denote by $D_v$ the circuit given by the maximal the sub-circuits of $D$ rooted in $v$ and whose output is $v$. If $v$ is an $\wedge$-gate, it is said {\em decomposable} if for every $v_1,v_2$ that are distinct inputs of $v$, it holds that $\var(D_{v_1}) \cap \var(D_{v_2}) = \emptyset$. An $\NNF$ circuit is in {\em Decomposable Normal Form} if all its $\wedge$-gates are decomposable. We will refer to such circuits as $\DNNF$. It is easy to see that one can find a satisfying assignment of a $\DNNF$ $D$ in time $O(\size(D))$. Moreover, if $D$ is a $\DNNF$ on variables $X$, $Y \subseteq X$ and $\tau \in \{0,1\}^Y$, then $D[\tau]$ is computed by a $\DNNF$ smaller than $D$ since we can plug the values of literals in $Y$ in the circuit $D$. 

\paragraph{Deterministic and Decision DNNF.} Let $D$ be a $\DNNF$. An $\vee$-gate in $D$ is called {\em deterministic} if for every $v_1,v_2$ that are distinct inputs of $v$, it holds that $D_{v_1} \land D_{v_2} \equiv 0$. $D$ is said deterministic if all its $\vee$-gates are deterministic. Observe that determinism is a semantic condition and is hard to decide from the $\DNNF$ only. In this paper, we will be mostly interested in {\em decision} gates that are a special case of deterministic gates. An $\vee$-gate $v$ of $D$ is a {\em decision gate} if it is binary and if there exists a variable $x$ and two gates  $v_1, v_2$ of $D$ such that $v$ is of the form $(x \land v_1) \lor (\neg x \land v_2)$. A {\em decision $\DNNF$}, $\decDNNF$ for short, is a $\DNNF$ for which every $\vee$-gate is a decision gate. It is easy to see that a $\decDNNF$ is deterministic. Figure~\ref{fig:dnnf} depicts a $\DNNF$. The output is represented by a square and the $\DNNF$ computes the boolean function $(\neg x \land z) \lor (x \land (y \lor z))$. It is easy to check that both $\land$-gates are decomposable. The output gate is also a decision gate on variable $x$. The other $\vee$-gate is not deterministic since the boolean function $y \land z$ is satisfiable.

\begin{figure}
  \centering
  \begin{tikzpicture}[scale=0.5, every node/.style={scale=0.9}]
\node[draw] (r) at (0,0) {$\vee$};
\node (a1) at (-1,-1) {$\wedge$};
\node (a2) at (1,-1) {$\wedge$};
\node (r2) at (0,-2) {$\vee$};

\node (x) at (-2,-2) {$x$};
\node (nx) at (2,-2) {$\neg x$};

\draw (x) -- (a1);
\draw (nx) -- (a2);

\node (y) at (-1,-3) {$y$};
\node (z) at (1,-3) {$z$};
\draw (y) -- (r2) -- (z);
\draw (r) -- (a1) -- (r2); \draw (z) -- (a2) -- (r);
\end{tikzpicture}
  \begin{tikzpicture}[scale=0.5]
    \node (y) at (-1,-1) {$y$};
    \node (z) at (1,-1) {$z$};
    \node (x) at (3,-1) {$x$};
    \draw (y) -- (0,0) -- (z);
    \draw (0,0) -- (1,1) -- (x);
  \end{tikzpicture}
  \caption{A $\DNNF$ and a vtree}
  \label{fig:dnnf}
\end{figure}
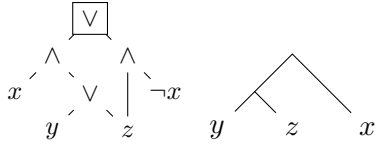

\paragraph{Structuredness.} Structuredness is a constraint on the way variables can be partitioned by a $\DNNF$. It may be seen as a generalization to trees of the variable ordering that is sometimes required in data structures such as OBDD~\cite{Wegener00} and was introduced in~\cite{PipatsrisawatD08}. Let $D$ be a $\DNNF$ on variables $X$. A {\em vtree} $T$ on $X$ is a binary tree whose leaves are in one-to-one correspondence with $X$. An $\land$-gate $v$ of $D$ {\em respects} a vertex $t$ of $T$ if it has exactly two inputs $v_1,v_2$ and if $\var(D_{v_1}) \subseteq X_{t_1}$ and  $\var(D_{v_2}) \subseteq X_{t_2}$ where $t_1,t_2$ are the children of $t$ in $T$ and $X_{t_1}$ (resp. $X_{t_2}$) is the set of variables that appears in the leaves of $T_{t_1}$ (resp. $T_{t_2}$). A $\DNNF$ $D$ respects a vtree $T$ if for every $\land$-gate $v$ of $D$, there exists a vertex $t$ of $T$ such that $v$ respects $t$. A $\DNNF$ $D$ is {\em structured} if there exists a vtree $T$ such that $D$ respects $T$. It can be checked that the $\DNNF$ depicted in Figure~\ref{fig:dnnf} respects the vtree given on the same figure.

\subsection{Structuredness and Branch Decomposition}
\label{sec:sbd}

In this section, we explain how most of the structure-based algorithms for $\sSAT$ work and how we can relate this to the fact that they are implicitly constructing a structured $\DNNF$ equivalent to the input formula.

The current techniques for solving $\sSAT$ by exploiting the structure of the input are all based on the same technique: they start by computing a ``good'' branch decomposition $T$ of the incidence graph of the formula $F$. Each vertex $t$ of the branch decomposition is then used to define a sub-formula $F_t$ and partial assignments $a_1,\dots,a_k$ of its variables. The number of solutions of $F_t[a_i]$ is then computed by dynamic programming along the branch decomposition in a bottom-up fashion. In all algorithms, the variables of $F_t$ are the variables of $F$ that label the leaves of $T_t$. The number of solutions of $F_t$ on some partial assignment $a_i$ is computed by multiplying and summing the number of solutions of $F_{t_1}$ and of $F_{t_2}$, where $t_1,t_2$ are the children of $t$ on restrictions of $a_i$ to the variables of $F_{t_1}$ and $F_{t_2}$ respectively. Those multiplications can be seen as a decomposable $\land$-gate and the sums can be seen as deterministic $\lor$-gates. Thus, the underlying $\DNNF$ constructed by those algorithms is naturally structured along the vtree obtained from the branch decomposition $T$ by forgetting the leaves that are labeled by clauses of the formula.

In this paper, we will thus say that a class of formula can be solved by using the {\em standard (dynamic programming) technique} if it can be compiled into deterministic structured $\DNNF$. The most general known algorithm exploiting the structure of the input, that we will call, from the author names, the {\em STV-algorithm}~\cite{SaetherTV14}, uses exactly this technique. It has been observed in~\cite{BovaCMS15} that this algorithm is actually implicitly constructing a deterministic structured $\DNNF$ equivalent to the input $\CNF$-formula, which reinforces the idea that the notion of structuredness captures the essence of the standard technique for solving $\sSAT$.  

\section{Compilation of $\beta$-acyclic formulas into $\decDNNF$}
\label{sec:comp}

We show how to construct a linear size $\decDNNF$ equivalent to a given $\beta$-acyclic formula $F$ (Theorem~\ref{thm:compileexists}). 
We use a dynamic programming approach by iteratively constructing $\decDNNF$ for subformulas of $F$. These subformulas are defined using general remarks on the structure of $\beta$-acyclic hypergraphs. 

\subsection{Structure of $\beta$-acyclic hypergraphs.}
\label{sec:structbh}
In this section, we fix a $\beta$-acyclic hypergraph $\calH$ with $n$ vertices and a $\beta$-elimination order $(x_1,\dots,x_n)$ of its vertices denoted by $<$. We denote by $<_\calH$ the order on $\calH$ defined as the lexicographical order on $\calH$ where $e \in \calH$ is seen as the $\{0,1\}^n$-vector $\vec{e}$ such that $\vec{e}_i = 1$ if $x_{n-i} \in e$ and $\vec{e}_i = 0$ otherwise. In other words, $e <_\calH f$ if and only if $\max (e \Delta f) \in f$.

From these orders, we construct a family of subhypergraphs of $\calH$ which will be interesting for us later. Let $x \in V$ and $e \in \calH$. We denote by $\Vle{x} = \{y \in V \mid y \leq x\}$. $\Vlt{x}$, $\Vge{x}$ and $\Vgt{x}$ are defined similarly. We denote by $\calH_e^x$ the subhypergraph of $\calH$ that contains the edges $f \in \calH$ such that there is a walk from $f$ to $e$ that goes only through edges smaller than $e$ and vertices smaller than $x$. 

Observe that, by definition, $\calH_e^x$ is a connected subhypergraph of $\calH$, with $e \in \calH_e^x$ and for all $f \in \calH_e^x$, $f \leq_\calH e$. Observe also that even if there is a walk from $f \in \calH_e^x$ to $e$ that goes only through vertices smaller than $x$, $f$ may hold vertices that are bigger than $x$. We insist on the fact that the whole edge $f$ is in $\calH_e^x$ and not only its restriction to $\Vle{x}$. 

We start by giving an example. Let $\calH = \{\{1,2\}, \{3,4\}, \{2,5\}, \{4,5\}, \{2,4,5\}\}$ be the hypergraph depicted on Figure~\ref{fig:hexdec}. One can easily check that $1 < 2 < 3 < 4 < 5$ is a $\beta$-elimination order and that the order $<_\calH$ is the following: $e_1 = \{1,2\} <_\calH e_2 = \{3,4\} <_\calH e_3 = \{2,5\} <_\calH e_4 = \{4,5\} <_\calH e_5 = \{2,4,5\}$. $\calH_{e_5}^4$ is the whole hypergraph since one can reach any edge from $e_5$ by going through vertices smaller than $4$. $\calH_{e_5}^3$ however is lacking the edge $e_2 = \{3,4\}$ since the only way of reaching $e_2$ from $e_5$ is to go through the vertex $4$ which is not allowed.

Observe that these subhypergraphs are naturally ordered by inclusion:
\begin{lemma}
  \label{lem:hexinc} Let $x,y \in V(\calH)$ such that $x \leq y$ and $e,f \in \calH$ such that $e \leq_\calH f$ and $V(\calH_e^x) \cap V(\calH_f^y) \cap \Vle{x} \neq \emptyset$. Then $\calH_e^x \subseteq \calH_f^y$. In particular, for all $y$, if $e \in \calH_f^y$ then $\calH_e^y \subseteq \calH_f^y$.
\end{lemma}
\begin{proof}
  Let $z \in V(\calH_e^x) \cap V(\calH_f^y) \cap \Vle{x}$ and let $g_1 \in \calH_e^x$ and $g_2 \in \calH_f^y$ be such that $z \in g_1 \cap g_2$. By definition, there exists a walk $\calP_1$ from $f$ to $g_2$ going through vertices smaller than $y$ and edges smaller than $f$ and a walk $\calP_2$ from $g_1$ to $e$ going through vertices smaller than $x$ and edges smaller than $e$. Since $z \leq x \leq y$ and $e \leq_\calH f$, $\calP = (\calP_1, z, \calP_2)$ is a walk from $f$ to $e$ going through edges smaller than $f$ and vertices smaller than $y$, that is $e \in \calH_f^y$. Now let $h \in \calH_e^x$ and let $\calP_3$ be a path from $e$ to $h$ going through vertices smaller than $x$ and edges smaller than $e$. Then $(\calP,\calP_3)$ is a walk from $f$ to $h$ going through vertices smaller than $y$ and edges smaller than $f$. That is $h \in \calH_f^y$, so $\calH_e^x \subseteq \calH_f^y$.
\end{proof}

We now state the main result of this section. Theorem~\ref{thm:hexvar} relates the variables of  $\calH_e^x$ to those of $\Vge{x}$ and $e$. This is crucial for the dynamic programming scheme of our algorithm:
\begin{theorem}
  \label{thm:hexvar} For every $x \in V$ and $e \in \calH$, $V(\calH_e^x) \cap \Vge{x} \subseteq e$.
\end{theorem}

In order to prove Theorem~\ref{thm:hexvar}, we need two easy intermediate lemmas:
\begin{lemma}
  \label{lem:ordervar} Let $e,f \in \calH$ such that there exists $x \in e \cap f$. If $e <_\calH f$ then $e \cap \Vge{x} \subseteq f$.
\end{lemma}
\begin{proof}
By definition of $\beta$-acyclic elimination order, we must have either $e \cap \Vge{x} \subseteq f \cap \Vge{x}$ or $f \cap \Vge{x} \subseteq e \cap \Vge{x}$. Now since $e <_\calH f$, we have $m = \max(e \Delta f) \in f$. If $m \leq x$, we have $e \cap \Vge{x} = f \cap \Vge{x}$. Otherwise, we have $e \cap \Vge{x} \subseteq f \cap \Vge{x}$ since $m \in (f \setminus e) \cap \Vge{x}$.
\end{proof}

A path $\calP = (e_0,x_0,\dots,x_{n-1},e_n)$ in $\calH$ is called {\em decreasing} if for all $i$, $e_i >_\calH e_{i+1}$ and $x_i > x_{i+1}$.
\begin{lemma}
  \label{lem:pathdecrease} For every $x \in V$, $e \in \calH$ and $f \in \calH_e^x$, there exists a decreasing path from $e$ to $f$ going through vertices smaller than $x$.
\end{lemma}
\begin{proof}
  By definition of $\calH_e^x$, there exists a path $\calP = (e_0,x_0,\dots,x_{n-1},e_n)$ with $e_0 = e$ and $e_n = f$ such that for all $i \leq n$, $e_i \leq_\calH e$ and $x_i \leq x$. We show that if $\calP$ is a shortest path among those going through vertices smaller than $x$, then it is also decreasing. Assume toward a contradiction that $\calP$ is a non-decreasing such shortest path. Remember that by definition of paths, the edges $(e_i)$ are pairwise distinct. The same is true for the vertices $(x_i)$.  Moreover, observe that since $\calP$ is a shortest path, then it holds that:
\begin{equation}
 \label{quote:short}
  \forall i < n \forall j \notin \{i,i+1\}, x_i \notin e_j.
  \tag{$\star$}
\end{equation}
Indeed, if there exists $i$ and $j \notin \{i,i+1\}$ such that $x_i \in e_j$, $\calP$ could be shortened by going directly from $e_i$ to $e_j$ if $j > i+1$ or from $e_j$ to $e_{i+1}$ if $j < i$.

Let $i = \min \{j \mid x_{j+1} > x_j \text{ or } e_{j+1} >_\calH e_j \}$ be the first indices where $\calP$ does not respect the decreasing condition, which exists if $\calP$ is not decreasing by assumption.

First assume $i = 0$. By definition of $\calP$, $e_0 = e >_\calH e_1$. Thus it must be that $x_0 < x_1$. By definition, $x_0 \in e_0 \cap e_1$ and by Lemma~\ref{lem:ordervar}, $e_1 \cap \Vge{x_0} \subseteq e_0$. Since $x_1 > x_0$, $x_1 \in e_1 \cap \Vge{x_0}$, thus $x_1 \in e_0$ which contradicts~(\ref{quote:short}).

Now assume $i > 0$. First, assume that $e_{i+1} >_\calH e_i$. By definition of $\calP$, it holds that $x_i \in e_i \cap e_{i+1}$ and then by Lemma~\ref{lem:ordervar}, $e_i \cap \Vge{x_i} \subseteq e_{i+1}$. Now observe that by minimality of $i$, $x_{i-1} > x_i$. Since $x_{i-1} \in e_i$, $x_{i-1} \in e_i \cap \Vge{x_i} \subseteq e_{i+1}$, which contradicts~(\ref{quote:short}).

Otherwise, $e_i >_\calH e_{i+1}$ and then $x_{i+1} > x_i$. By Lemma~\ref{lem:ordervar} again, $e_{i+1} \cap \Vge{x_i} \subseteq e_i$. Since $x_{i+1} \in e_{i+1}$, this implies that $x_{i+1} \in e_{i+1} \cap \Vge{x_i} \subseteq e_i$, which contradicts~(\ref{quote:short}). It follows that such $i$ does not exist, that is, $\calP$ is decreasing.
\end{proof}

\begin{proof}[Proof (of Theorem~\ref{thm:hexvar}).]
We show by induction on $n$ that for any decreasing path $\calP = (e_0,x_0,\dots, e_n)$  from $e_0$ to $e_n$, we have $e_0 \supseteq e_n \cap \Vge{x_0}$. If $n = 0$, then $e_n = e_0$ and the inclusion is obvious. Now, let $\calP = (e_0,x_0,\dots,e_n,x_n,e_{n+1})$. By induction, $e_0 \supseteq e_n \cap \Vge{x_0}$ since $(e_0,x_0,\dots,e_n)$ is a decreasing path from $e_0$ to $e_n$. Now by Lemma~\ref{lem:ordervar}, since $x_n \in e_{n+1} \cap e_n$ and $e_{n+1}<_\calH e_n$, we have $e_{n+1} \cap \Vge{x_n} \subseteq e_n$. Since $x_0>x_n$, $e_{n+1} \cap \Vge{x_0} \subseteq e_{n+1} \cap \Vge{x_n} \subseteq e_n$. Thus $e_{n+1} \cap \Vge{x_0} \subseteq e_n \cap \Vge{x_0} \subseteq e_0$ which concludes the induction.

Now let $e \in \calH$, $x \in V(\calH)$ and $f \in \calH_e^x$. By Lemma~\ref{lem:pathdecrease}, there exists a decreasing path from $e$ to $f$ going through vertices smaller than $x$. From what precedes, $f \cap \Vge{x} \subseteq e$. Therefore $V(\calH_e^x) \cap \Vge{x} \subseteq e$.
\end{proof}

\subsection{Constructing the $\decDNNF$.}
\label{sec:exdecdnnf}

Given a $\CNF$-formula $F$ with hypergraph $\calH$, we can naturally define a family of subformulas $F_e^x$ from $\calH_e^x$ as the conjunction of clauses corresponding to the edges in $\calH_e^x$, that is $F_e^x = \{C \in F \mid \var(C) \in \calH_e^x\}$.  Theorem~\ref{thm:hexvar} implies in particular that $\var(F_e^x) \subseteq (e \cup \Vlt{x})$. Thus, if $\tau$ is an assignment of variables $(e \cap \Vgt{x})$, then $F_e^x[\tau]$ has all its variables in $\Vle{x}$. We will be particularly interested in such assignments: for a clause $C \in F$, denote by $\tau_C$ the only assignment of $\var(C)$ such that $\tau_C \not \models C$ and by $\tau_C^x := \tau_C |_{\Vgt{x}}$. We construct a $\decDNNF$ $D$ by dynamic programming such that for each clause $C$ with $\var(C) = e$ and variable $x \in V$, there exists a gate in $D$ computing $F_e^x[\tau_C^x]$, which is a formula with variables in $\Vle{x}$. Lemma~\ref{lem:fexplode} and Corollary~\ref{cor:fexplode} describe everything needed for the dynamic programming algorithm by expressing $F_e^x$ as a decomposable conjunction of precomputed values.

\begin{lemma}
\label{lem:fexplode}
  Let $x \in \var(F)$ such that $x \neq x_1$ and let $y \in \var(F)$ be the predecessor of $x$ for $<$. Let $e \in \calH(F)$ and $\tau: (e \cap V_{\geq x}) \rightarrow \{0,1\}$. Then either $F_e^x[\tau] \equiv 1$ or there exists $U \subseteq \calH_e^x$ and for all $g \in U$ a clause $C(g) \in F_e^x$ with $\var(C(g)) = g$ such that \[F_e^x[\tau] \equiv \bigwedge_{g \in U} F_{g}^y[\tau_{C(g)}^y].\]
Moreover, this conjunction is decomposable and $U$ can be found in polynomial time in $\size(F)$.
\end{lemma}
\begin{proof}
  Assume first that for all $C \in F_e^x$, $\tau \models C$. Thus $F_e^x[\tau] \equiv 1$ since every clause of $F_e^x$ is satisfied by $\tau$.

  Now assume that there exists $C \in F_e^x$ is such that $\tau \not \models C$. This means that $\tau \simeq \tau_C$.   We let $A = \{\var(C) \mid C \in F_e^x \text{ and } \tau \not \models C\} \neq \emptyset$ by assumption. Observe that \[F_e^x[\tau] \equiv \bigwedge_{\substack{C \in F_e^x \\ \var(C) \in A}} C[\tau] \] since for every $C \in F_e^x$, if $\var(C) \notin A$, $\tau \models C$ by construction of $A$.

Let $U = \{g \in A \mid \forall f \in A \setminus \{g\}, g \notin \calH_{f}^y \}$. For each $g \in U$, we choose an arbitrary clause $C(g)$ such that $\var(C(g)) = g$ and $\tau \not \models C(g)$. Such a clause exists since $U \subseteq A$. We claim that $U$ meets the conditions given in the statement of the lemma. 

We start by observing that $U$ can be computed in polynomial time in $\size(F)$. Indeed, computing $F_e^x$ for all $e,x$ can be done in polynomial time as it boils down to a computation of connected component in a hypergraph. Now to compute $A$, it is enough to test for every $C \in F_e^x$ that $\tau \not \models C$ which can be done in polynomial time in $\size(F)$. Finally, extracting $U$ from $A$ can also be done in polynomial time by testing for every $g \in A$ if $g$ respects the given condition: it is enough to test for every $f \in A \setminus \{g\}$ if $g \notin \calH_f^y$, which is possible since we can compute $\calH_f^y$ easily.

Now let $f \in A$. We show that there exists $g \in U$ such that $f \in \calH_g^y$. If $f \in U$, then we are done since $f \in \calH_f^y$.  Now assume that $f \notin U$. By definition of $U$, $B = \{g \in A \setminus \{f\} \mid f \in \calH_{g}^y \} \neq \emptyset$. We choose $g$ to be the maximum of $B$ for $\leq_\calH$. We claim that $g \in U$. Indeed, assume there exists $g' \in A$ such that $g \in \calH_{g'}^y$ and $g < g'$. By Lemma~\ref{lem:hexinc},  $\calH_g^y \subseteq \calH_{g'}^y$ and since $f \in \calH_g^y$, we also have $f \in \calH_{g'}^y$, that is, $g' \in B$. Yet, $g = \max(B)$ and $g \leq g'$, that is, $g = g'$. Thus $g \in U$. 

We thus have proved that for all $f \in A$, there exists $g \in U$ such that $f \in \calH_g^y$. Thus if $C$ is a clause of $F_e^x$, either $\var(C) \notin A$ and then $\tau \models C$ by definition of $A$, or $\var(C) \in A$, then there exists $g \in U$ such that $\var(C) \in \calH_g^y$, that is, $C \in F_g^y$. Now, if $C \in F_g^y$ for some $g \in U$, then $C \in F_e^x$ too since by Lemma~\ref{lem:hexinc}, $F_g^y \subseteq F_e^x$. Thus 
\[ F_e^x[\tau] \equiv \bigwedge_{g \in U} F_g^y[\tau]. \]

Let $g \in U$. We show that $\tau|_{\var(F_g^y)} = \tau_{C(g)}^y$. Observe that by Theorem~\ref{thm:hexvar}, $\var(F_g^y) \cap \Vge{x} = V(\calH_g^y) \cap \Vge{x} \subseteq g \cap \Vge{x}$. Since $\tau$ assigns variables from $e \cap \Vge{x}$: 
\[
\begin{aligned}
  \tau|_{\var(F_g^y)} & = \tau |_{\var(F_g^y) \cap \Vge{x} \cap e} \\
  & = \tau |_{g \cap \Vge{x} \cap e}
\end{aligned}
\]
 Moreover,  since $g \in \calH_e^x$, we have $g \cap \Vge{x} \subseteq e \cap \Vge{x}$ by Theorem~\ref{thm:hexvar} again. Thus $g \cap \Vge{x} \cap e = g \cap \Vge{x}$. In other words, $\tau|_{\var(F_g^y)}  = \tau |_{g \cap \Vge{x}}$.

Since $\tau$ assigns all variables of $e \cap \Vge{x}$ by assumption, $\tau |_{g \cap \Vge{x}}$ assigns all variables of $g \cap \Vge{x}$. Finally, since $\tau \not \models C(g)$ by construction of $C(g)$, we have $\tau \simeq \tau_{C(g)}^y$. Since by definition $\var(C(g)) = g$, it follows that $\tau|_{\var(F_g^y)} = \tau_{C(g)}^y$. So far, we have proven that
\[ F_e^x[\tau] \equiv \bigwedge_{g \in U} F_g^y[\tau_{C(g)}^y]. \]

It remains to show that this conjunction is decomposable, that is, for all $g_1, g_2 \in U$, 
$\var(F_{g_1}^y[\tau_{C(g_1)}^y]) \cap \var(F_{g_2}^y[\tau_{C(g_2)}^y]) = \emptyset$. Let $g_1,g_2 \in U$ with $g_1 <_\calH g_2$ and assume there exists $z \in \var(F_{g_1}^y[\tau_{C(g_1)}^y]) \cap \var(F_{g_2}^y[\tau_{C(g_2)}^y])$, that is, $z \in \var(F_{g_1}^y) \cap \var(F_{g_2}^y) \cap \Vle{y}$. From what precedes, $\tau$ assigns every variable of $F_{g_1}^y$ greater than $x$. By Lemma~\ref{lem:hexinc}, we have $F_{g_1}^y \subseteq F_{g_2}^y$, which contradicts the fact that $g_1 \in U$.
\end{proof}

\begin{corollary}
  \label{cor:fexplode} Let $x \in \var(F)$ such that $x \neq x_1$ and let $y \in \var(F)$ be the predecessor of $x$ for $<$.  For every $C \in \calH(F)$, there exist $U_0,U_1 \subseteq \calH_{\var(C)}^x$ and for all $g \in U_0 \cup U_1$ a clause $C(g) \in F_{\var(C)}^x$ with $\var(C(g)) = g$ such that \[F_{\var(C)}^x[\tau_C^x] \equiv (x \land \bigwedge_{g \in U_1} F_{g}^y[\tau_{C(g)}^y]) \lor (\neg x \land \bigwedge_{g \in U_0} F_{g}^y[\tau_{C(g)}^y]).\]
Moreover, all conjunctions are decomposable and $U_0,U_1$ can be found in polynomial time in $\size(F)$.
\end{corollary}
\begin{proof}
  Let $\tau_1 = \tau_C^x \cup \{x \mapsto 1 \}$ and $\tau_0 = \tau_C^x \cup \{x \mapsto 0 \}$. We observe that
\[ F_{\var(C)}^x[\tau_C^x] = (x \land F_{\var(C)}^x[\tau_1]) \lor (\neg x \land F_{\var(C)}^x[\tau_0]).\]
Clearly, $x \notin \var(F_{\var(C)}^x[\tau_1])$ and $x \notin \var(F_{\var(C)}^x[\tau_0])$, thus, both conjunctions are decomposable. Now, applying Lemma~\ref{lem:fexplode} on $F_{\var(C)}^x[\tau_0]$ and on $F_{\var(C)}^x[\tau_1]$ yields the desired decomposition.
\end{proof}

\begin{theorem}
  \label{thm:compileexists} Let $F$ be a $\beta$-acyclic $\CNF$-formula. One can construct in polynomial time in $\size(F)$ a $\decDNNF$ $D$ of size $O(\size(F))$ and fanin at most $|\calH|$ computing $F$.
\end{theorem}
\begin{proof}
  Let $\calH$ be the hypergraph of $F$ and $<$ a $\beta$-elimination order. Let $\var(F) = \{x_1, \dots, x_n\}$ where $x_i < x_j$ if and only if $i<j$. We construct by induction on $i$ a $\decDNNF$ $D_i$ of fanin $|\calH|$ at most such that for each $e \in \calH$, $C \in F$ such that $\var(C) = e$ and $j \leq i$, there exists a gate in $D_i$ computing $F_e^{x_j}[\tau_C^{x_j}]$ and $|D_i| \leq 7 \cdot (\sum_{j=1}^i c(x_j))$ where $c(x_j)$ is the number of clauses in $F$ holding $x_j$.

We start by explaining how $D_1$ is constructed. Let $e \in \calH$. If $x_1 \notin e$, then $F_e^{x_1}$ contains only the clauses $C$ such that $e = \var(C)$. For such a $C$, $\tau_C^{x_1} = \tau_C$, thus $F_e^{x_1}[\tau_C] = 0$. Now, if $x_1 \in e$, $F_e^{x_1}$ contains only clauses $D$ such that $x_1 \in \var(D) \subseteq e$ since $x_1$ is the first element of the elimination order. Let $C$ be a clause such that $\var(C) = e$. For every $D \in F_e^{x_1}$, $\var(D) \subseteq \var(C)$, thus $F_e^{x_1}[\tau_C^{x_1}]$ has only one variable: $x_1$. Thus $F_e^{x_1}[\tau_C^{x_1}]$ is equivalent to either $x_1$, $\neg x_1$ or $0$. We thus define $D_1$ to be the $\decDNNF$ with at most three gates $x_1, \neg x_1$ and $0$, which are input gates. We have $|D_1| \leq 3 \leq 7 \cdot c(x_1)$.

Now lets assume $D_i$ is constructed. To ease notations, let $x = x_{i+1}$. Let $e \in \calH$ and $C$ be a clause such that $\var(C) = e$. We want to add a gate in $D_i$ that will compute $F_e^x[\tau_C^x]$. If $x \notin e$, then $\calH_e^x = \calH_e^{x_i}$ since by Theorem~\ref{thm:hexvar}, $\var(\calH_e^{x_i}) \subseteq (e \cup \Vlt{x_i})$. Thus $F_e^x = F_e^{x_i}$ and $\tau_C^x = \tau_C^{x_i}$. Therefore, there is already a gate computing $F_e^x[\tau_C^x]$ in $D_i$.

Assume now that $x \in e$. By Corollary~\ref{cor:fexplode}, we can compute $F_{\var(C)}^x[\tau_C^x]$ for every $C$ with $\var(C) = e$ by adding at most one decision-gate and a fanin $|\calH|$ decomposable and-gate to $D_i$ since for all values appearing in the statement of Corollary~\ref{cor:fexplode} there exists a gate in $D_i$ computing it. Moreover such gate can be found in polynomial time. That is, we add to $D_i$ at most $7$ gates to compute $F_{\var(C)}^x[\tau_C^x]$. We have to do this for each $C \in F$ such that $x \in \var(C)$. We thus add at most $7c(x)$ gates in $D_i$. Thus $|D_{i+1}| \leq 7 \cdot \sum_{j\leq i+1} c(x_j)$.

To conclude, assume that $\calH$ is connected and let $e = \max(\calH)$. We have $\calH_e^{x_n} = \calH$ since there is a path from $e$ to every other edge in $\calH$. Thus $F_e^{x_n} = F$. Let $C$ be a clause with $\var(C) = e$. The assignment $\tau_C^{x_n}$ is empty, thus $F_e^{x_n}[\tau_C^{x_n}] \equiv F$. Hence, there is a gate in $D_n$ that computes $F$ and $D_n$ is of size at most $7 \cdot \size(F)$ and fanin $|\calH|$ at most. Each step can be done in polynomial time in $\size(F)$.

If $\calH$ is not connected, then each connected component of $\calH$ is $\beta$-acyclic, thus we can compile them independently and take the decomposable conjunction of these $\decDNNF$.
\end{proof}

We conclude this section by giving insights on the significance of Theorem~\ref{thm:compileexists} from a practical point of view. Most practical tools for $\sSAT$ are based on an algorithm called exhaustive DPLL with caching~\cite{HuangD05,sang04,thurley06,BacchusDP03} which works as follows: given $F$, the algorithm starts by trying to write $F$ as $F_1 \land F_2$ with $F_1$ and $F_2$ having no common variables. If it succeeds, it computes recursively $\#F_1$, $\#F_2$ and returns $\#F_1 \cdot \#F_2$. Otherwise, it chooses a variable $x$ and returns $\#F[x \mapsto 0]+\#F[x \mapsto 1]$. In addition, these tools use caching techniques to avoid redoing the same computation twice. It was observed in~\cite{HuangD05} that the trace of such algorithms is exactly a $\decDNNF$. It is not hard to see that the construction given in Theorem~\ref{thm:compileexists} is the trace of a run of an exhaustive DPLL algorithm where the variables are chosen in a reverse $\beta$-elimination order. This shows that if the right elimination order of the variables is chosen (and this order can be computed greedily in polynomial time), then practical tools for solving $\sSAT$ can in theory solve $\beta$-acyclic formulas in polynomial time.

\section{Deviation from the technique based on branch decompositions}
\label{sec:deviation}

In this section, we finally prove that standard techniques based on branch decompositions fail on $\beta$-acyclic formulas. Recall that we have defined in Section~\ref{sec:sbd} the standard technique to be the implicit construction of a polynomial size structured $\DNNF$ equivalent to the input formula. We formally prove the following:

\begin{theorem}
  \label{thm:betalb} There exists an infinite family $\calF$ of $\beta$-acyclic $\CNF$-formulas such that for every $F \in \calF$ having $n$ variables, there is no structured $\DNNF$ of size less than $2^{\Omega(\sqrt{n})}$ computing $F$.
\end{theorem}

We use techniques based on communication complexity tools developed in~\cite{BovaCMS16} to prove lower bounds on the size of structured $\DNNF$.

\begin{definition}
  Let $r$ be a boolean function on variables $X$ and let $(Y,Z)$ be a partition of $X$. The function $r$ is a $(Y,Z)$-rectangle if and only if for every $\tau,\tau' \in \{0,1\}^X$ such that $\tau \models r$ and $\tau' \models r$, we have $(\tau|_Y \cup \tau'|_Z) \models r$. A $(Y,Z)$-rectangle cover of a boolean function $f$ is a set $R=\{r_1,\dots,r_q\}$ of $(Y,Z)$-rectangles such that $\sat(f) = \bigcup_{i=1}^q \sat(r_i)$.  
\end{definition}

\begin{theorem}[\cite{BovaCMS16},\cite{DarwicheP10}]
  \label{thm:rectlb} Let $D$ be a $\DNNF$ on variables $X$ respecting the vtree $T$. For every vertex $t$ of $T$, there exists a $(X_t,X \setminus X_t)$-rectangle cover of $D$ of size at most $|D|$, where $X_t = \var(T_t)$.
\end{theorem}

Given a $\CNF$-formula $F$, we define $\hat F$ to be the formula $\{K \cup \{c_K\} \mid K \in F\}$ on variables $\{c_K \mid K \in F \} \cup \var(F)$. Intuitively, $\hat F$ is the formula obtained by adding one fresh variable $c_K$ in each clause $K$ of $F$. Our main lower bound  relates  the incidence $\MIMw$ of a monotone $\CNF$-formula to the size of structured $\DNNF$ computing $\hat{F}$.

\begin{theorem}
\label{thm:mmwlb}
Let $F$ be a monotone formula of incidence $\MIMw$ $k$. Any structured $\DNNF$ computing $\hat{F}$ is of size at least $2^{k/2}$.
\end{theorem}

The proof of Theorem~\ref{thm:mmwlb} heavily relies on the following lower bound and on Theorem~\ref{thm:rectlb}:
\begin{lemma}
  \label{lem:matchlb} Let $X=\{x_1,\dots,x_k\}$ and $Y = \{y_1,\dots,y_k\}$ be two disjoint sets of $k$ variables. The number of $(X,Y)$-rectangles needed to cover the $\CNF$-formula $F = \bigwedge_{i=1}^k (x_i \lor y_i)$ is at least $2^k$.
\end{lemma}
\begin{proof}
Let $\{R_1, \dots, R_q\}$ be a $(X,Y)$-rectangle cover of $F$. For $K \subseteq \{1,\dots,k\}$, we denote by $\tau_K$ the assignment such that $\tau_K(x_i) = 1$ if $i \in K$ and $0$ otherwise and $\tau_K(y_i) = 1-\tau_K(x_i)$. Observe that by definition, for every $K \subseteq \{1,\dots,k\}$, $\tau_K \models F$. We claim that if $\tau_K \models R_i$ then for any $K' \neq K$, we have $\tau_{K'} \not \models R_i$. For the sake of contradiction, assume there exist $K,K'$ such that $K \neq K'$, $\tau_K \models R_i$ and $\tau_{K'} \models R_i$. Without loss of generality, we can assume that there exists $i \in K \setminus K'$. By definition of rectangles, $\tau' = \tau_{K'}|_{X} \cup \tau_{K}|_Y \models R_i$. But $\tau'(x_i) = \tau'(y_i) = 0$ and then $\tau' \not \models F$ which contradicts the definition of $R_i$. Since there are $2^k$ different subsets of $\{1,\dots,k\}$ and each $\tau_K$ satisfies disjoint rectangles, we have that $q \geq 2^k$.
\end{proof}

\begin{proof}[Proof (of Theorem~\ref{thm:mmwlb}).] Let $G = \Ginc(F)$ and $D$ be a structured $\DNNF$ computing $\hat{F}$. We claim that $|D| \geq 2^{k/2}$. 

Let $T$ be the vtree respected by $D$. Observe that the variables of $\hat{F}$ are in one to one correspondence with $V(G)$ thus we can see $T$ as a branch decomposition of $G$. Since $G$ is of $\MIMw$ $k$, there exists a vertex $t$ of $T$ such that there is an induced matching $M = \{(x_1,y_1), \dots, (x_q,y_q)\}$ with $q \geq k$ in $G[V_t, V(G) \setminus V_t]$ where $V_t$ denotes the labels of the leaves of $T_t$. Let $e = (x,y)$ be an edge of $M$. Since it is an edge of $G$, too, one end point of $e$ corresponds to a variable $x_e$ of $F$ and the other to a clause $c_e \in F$. Let $M'$ be the set of edges $e$ of $M$ such that $x_e \in V_t$ and $c_e \notin V_t$ and let $M''$ be the set of edges $e$ of $M$ such that $x_e \notin V_t$ and $c_e \in V_t$. It is readily verified that $M = M' \uplus M''$. Let $N$ be the largest of these two sets. $N$ is thus an induced matching of $G[V_t,V(G) \setminus V_t]$ of size at least $k/2$. Moreover,  if $e,e' \in N$ are distinct, we have $x_{e'} \notin c_e$. Indeed, if $x_{e'} \in c_e$ then they are connected by an edge of $G$ and this edge is across $V_t$ and $V(G) \setminus V_t$ by construction of $N$. Thus, if such an edge exists, it violates the assumption that $N$ is an induced matching of $G[V_t, V(G) \setminus V_t]$.

Now let $\tau$ be the following partial assignment of $\var(\hat{F})$: if $C$ is a clause that does not appear in $N$, we let  $\tau(C) = 1$. If $x$ is a variable of $F$ that does not appear in $N$,  we let $\tau(x) = 0$. We claim that $\hat{F}[\tau] \equiv \bigwedge_{e \in N} (x_e \lor c_e)$. Indeed, each clause $C$ that does not appear in $N$ is already satisfied in $\hat{F}[\tau]$ since $\tau(C) = 1$ and for the remaining clauses, the variables that do not appear in $N$ disappear as they are set to $0$ (remember that $F$ is monotone). Moreover, if $e,e' \in N$ are distinct edges of $N$, we have that $x_e \notin c_{e'}$ thus the only variables remaining in the clause $c_e$ is $x_e$ for each $e \in N$. 

Now since $\hat{F}$ is computed by $D$, $\hat{F}[\tau]$ is computed by $D' = D[\tau]$ which is a structured $\DNNF$ smaller than $D$. By Theorem~\ref{thm:rectlb}, there is a $(V_t,V(G) \setminus V_t)$-rectangle cover of $D'$ of size at most $\size(D')$ and by Lemma~\ref{lem:matchlb}, we need at least $2^{|N|} \geq 2^{k/2}$ rectangles to cover $F[\tau]$. Thus, $\size(D) \geq \size(D') \geq 2^{k/2}$.
\end{proof}

Theorem~\ref{thm:betalb} is a corollary of Theorem~\ref{thm:mmwlb} and Theorem~\ref{thm:betavsmm}:
\begin{proof}[Proof of Theorem~\ref{thm:betalb}] 
Let $F$ be a $\beta$-acyclic formula. We claim that $\hat{F}$ is also $\beta$-acyclic. Indeed, let $(x_1,\dots,x_n)$ be a $\beta$-elimination order for $\calH(F)$. We claim that $(c_1,\dots,c_m,x_1,\dots,x_n)$ is a $\beta$-elimination order of $\calH(\hat{F})$ where $c_1,\dots,c_m$ are the variables of $\hat{F}$ corresponding to the clauses of $F$. Indeed, for all $i$, $c_i$ is in exactly one edge of $\calH(\hat{F})$ and can thus be eliminated from the start. Finally, $\calH(\hat{F}) \setminus \{c_1,\dots,c_m\} = \calH(F)$, thus $(x_1,\dots,x_n)$ is a $\beta$-elimination order of $\calH(\hat{F}) \setminus \{c_1,\dots,c_m\}$.

To every hypergraph $\calH$, we can associate a monotone formula $\CNF(\calH)$ whose variables are the vertices of $\calH$ and clauses are the edges of $\calH$ without negations. It is readily verified that the hypergraph of $\CNF(\calH)$ is $\calH$. Let $\calG$ be the family of $\beta$-acyclic hypergraphs with $\MIMw$ of $\Omega(n)$ from Theorem~\ref{thm:betavsmm}  and let $\calF = \{\widehat{\CNF(\calH)} \mid \calH \in \calG \}$. From what precedes, $\calF$ is a family of $\beta$-acyclic hypergraphs and by Theorem~\ref{thm:mmwlb}, if $F \in \calF$ has $m$ clauses and $N = n+m$ variables then any structured $\DNNF$ computing $F$ is of size at least $2^{\Omega(n)}$. The statement of Theorem~\ref{thm:betalb} follows since the number of edges in a $\beta$-acyclic hypergraph with $n$ vertices is at most $n(n+1)/2$~(Remark 13 in \cite{BraultB14}). Thus, $N = O(n^2)$, {\em i.e.} $n = \Omega(\sqrt N)$.
\end{proof}

\section{Discussion}
\label{sec:conclusion}

We discuss here further directions that can be studied from the results presented in this paper. In Section~\ref{sec:deviation}, we have shown that $\beta$-acyclic formulas cannot be compiled into structured $\DNNF$ contrary to other known tractable classes. It would be interesting to study the opposite question, that is, to understand if classes tractable with respect to the STV-algorithm can be compiled into $\decDNNF$. A positive answer to this question would open interesting perspectives as it would imply that all known tractable structural restrictions for $\sSAT$ can be processed using exhaustive DPLL with caching, which could lead to a practical use of such theoretical result and to the design of interesting heuristic for the order in which variables are eliminated in DPLL based on structural restrictions. A negative answer would show that some ``easy'' cases are missed by practical tools and that it would be worth investing time to develop practical tools taking the formula structure into account.

Another direction is suggested by Theorem~\ref{thm:mmwlb} which says that the $\MIMw$ of the formula is closely related to the size of the smallest structured $\DNNF$ for $\hat F$. The most general graph parameter that is known to lead to polynomial time execution with the STV-algorithm is the $\MIMw$: $\sSAT$ can be solved in time $m^{\Omega(k)} \poly(n+m)$ on a formula with $m$ clauses, $n$ variables and of $\MIMw$ $k$. Theorem~\ref{thm:mmwlb} almost proves the optimality of such running time for compilation into structured $\DNNF$. 


\bibliography{biblio.bib}

\end{document}